\documentclass[a4paper]{article}

\usepackage[latin1]{inputenc}
\usepackage{subfigure}

\usepackage{epstopdf}

\usepackage{graphics}
\usepackage{epsfig}
\usepackage{amssymb,amsmath,graphicx}
\usepackage{url}
\usepackage{footnote}

\begin{document}

\author{
Serge Burckel
\footnote{ERMIT, Universit\'e de La R\'eunion, France}
\addtocounter{footnote}{2}
\and
Emeric Gioan
\footnote{LIRMM, CNRS,  France}
\addtocounter{footnote}{-2}
\and
Emmanuel Thom\'e
\footnote{LORIA, INRIA,  France}
}


  
  \title
{Computation with No Memory, \\ 
%
and Rearrangeable 
Multicast Networks%
%
%
}

\date{\footnotesize  \vspace{0.5cm} february 2014 preprint\\ to appear in Discrete Mathematics and Theoretical Computer Science}

  

 
  
 

%
%
%



\newenvironment{proof}{\noindent\textit{Proof: }}{{\hfill $\Box$}}
\newtheorem{lemma}{Lemma}
\newtheorem{theorem}[lemma]{Theorem}
\newtheorem{proposition}[lemma]{Proposition}
\newtheorem{corollary}[lemma]{Corollary}
\newtheorem{definition}[lemma]{Definition}
\newtheorem{example}[lemma]{Example}
\newtheorem{remark}[lemma]{Remark}
\newtheorem{open problem}{Open problem}

\def\bR{\ensuremath{\mathbb{R}}}
\def\bZ{\ensuremath{\mathbb{Z}}}
\def\bN{\ensuremath{\mathbb{N}}}

\renewcommand{\arraystretch}{0.7}
\renewcommand{\tabcolsep}{1pt}

\def\square{\hbox{\vrule\vbox{\hrule\phantom{o}\hrule}\vrule}}
\def\Min{\hbox{\rm Min\ }}

\def\gl{|}

\def\f{\psi}

\def\Bsn{B_{s,n}}
\def\B{B}
\def\R{\B^{-1}}

\def\trio{{$P$-factorisation }}

\font\ptirm=cmr10 scaled 700


\maketitle

\begin{abstract}
We investigate the computation of 
mappings from a set
$S^n$ to itself with {\it in situ programs}, that is using no extra
variables than the input, and performing modifications of one component at a~time, hence using no extra memory. 
In this paper, we survey this problem introduced in previous papers by the authors, we detail its close relation with rearrangeable multicast networks, and we provide new results for both viewpoints.

A bijective mapping can be computed by $2n-1$ component modifications, that is by a program of length  $2n-1$,
a result equivalent to the 
rearrangeability of the concatenation of two reversed butterfly networks.
For a general arbitrary mapping, we give two methods to build a program with maximal length
$4n-3$. 
Equivalently, this yields  rearrangeable 
multicast routing methods for
the network formed by four successive butterflies with alternating reversions.
The first method is available for any set $S$ and practically equivalent to a 
known method in network theory. The second method, a refinement of the first, described when $|S|$ is a power of~$2$, is new and allows more~flexibility than the known method.

For a linear mapping, when $S$ is any field, or a quotient of
an Euclidean domain (e.g.\ $\bZ/s\bZ$ for any integer $s$),
we build a program with maximal length 
$2n-1$. In this case the assignments are
also linear, thereby particularly efficient from the algorithmic viewpoint, 
and giving moreover directly a program for the inverse when it exists.
This yields also a new result on matrix decompositions, and a new result on the multicast properties of two successive reversed butterflies. Results of this flavour were known only for the boolean field $\bZ/2\bZ$.

\smallskip

\end{abstract}

\section{Introduction}
\label{sec:intro}

The mathematical definition of a mapping $E:S^n\rightarrow S^n$ can be thought of 
as the parallel computation of $n$ assignment mappings $S^n\rightarrow S$ performing the mapping $E$, either by modifying at the same time the $n$ component variables,
or mapping the $n$ input component variables onto $n$ separate output component variables.
If one wants to compute sequentially the mapping $E$ by modifying the components one by one and using no other memory than the input variables whose modified values overwrite the initial values, one necessarily needs to transform the $n$ mappings $S^n\rightarrow S$ in a suitable way. We call {\it in situ computation} this way of computing a mapping, and it turns out that it is always possible with a number of assignments linear with respect to $n$ and a small factor depending on the mapping type.
%

%
The idea of developing in situ computation came from a natural computation viewpoint:
transformation of an entity or structure is made inside this entity or structure, meaning with no extra space.
As a preliminary example, consider the mapping $E:S^2\rightarrow S^2$ defined by $E(x_1,x_2)=(x_2,x_1)$
consisting in the exchange of two variables for a group $S$. 
A basic program computing $E$ is: $x':=x_1$; $x_1:=x_2$; $x_2:=x'$.
An in situ program for $E$ avoids the use of the extra variable $x'$, with the successive assignments $x_1:=x_1 + x_2$;
$x_2:=x_1 - x_2$;
$x_1:=x_1 - x_2$.
%
In situ computation can be seen as a far reaching generalization of this classical computational trick. See Section \ref{sec:insitu} for a formal definition.

The problem of building in situ programs has already been introduced and considered under equivalent terms in 
\cite{Bu96}\cite{BuMo00}\cite{BuMo04}\cite{BuMo04bis}\cite{BuGi08}\cite{BuGiTh09}.
%
In the first papers, it had been proved
that in situ computations are always possible \cite{Bu96},
 that three types of assignments are sufficient to perform this kind of computations \cite{BuMo00},
that the length of in situ computations of mappings on $\{0,1\}^n$ is bounded by $n^2$ \cite{BuMo04},
and that any linear mapping on $\{0,1\}^n$ is
computed with $2n-1$ linear assignments \cite{BuMo04bis}. 
It turned out that, though this had not been noticed in those papers, this problem has close relations with the problem of finding
rearrangeable (non-blocking) multicast routing methods for multistage interconnection networks obtained by concatenations of butterfly networks  (see Section \ref{sec:MIN} for a formal definition).
Also, several existence results on in situ programs can be deduced from results in the network field.
%
%
This relation has been partially presented in \cite{BuGiTh09}, which proposed also improved bounds for mappings of various types on more general sets than the boolean set, 
and which can be considered as a preliminary conference version of the present paper 
(with weaker results, fewer detailed constructions,  fewer references and  fewer illustrations).  
%
Let us also mention \cite{BuGi08} which presented the subject to an electronics oriented audience for the sake of possible applications,  \cite{De11} which presented the subject to  non-specialists and general public,
and
\cite{Ga1} 
which gives some results in the continuation of~ours.

%
%
%
%

In the present paper, we survey the relation between in situ computations and multicast rearrangeable networks, through precise results and historical references.
Also, we recall that a bijective mapping can be computed by a program of length $2n-1$;
we give, for a general arbitrary mapping, two methods to build a program with maximal length
$4n-3$, one of which is equivalent to a known method in network theory (Section \ref{sec:general}), one of which    available on the boolean set is new and more flexible (Section \ref{sec:boolean});
and we build, for a linear mapping of a rather general kind, a program with maximal length $2n-1$ (Section \ref{sec:linear}). Moreover, we end each  main section with an open problem. Our techniques use combinatorics and modular arithmetic.
Finally,  our aim is to give a precise and illustrated survey on this subject, for both viewpoints (computations and networks), so as to sum up old and new available results in a unified 
appropriate 
framework. 

\bigskip

Let us first detail some links with references from the network viewpoint. Multistage interconnection networks have been an active research area
over the past forty years.  We refer the reader to \cite{Hw04}\cite{Lei92} for background on this field. 
Here, 
an assignment,
which is a mapping $S^n\rightarrow S$,
 is regarded as a set of edges in a bipartite graph between $S^n$ (input) and $S^n$ (output)
where an edge corresponds to the modification of the concerned component.
See Section \ref{sec:MIN} for details.
All the results of the paper
can be translated in this context. More precisely, making successive
modifications of all consecutive components of $X\in S^n$ 
is equivalent to
routing a butterfly network (i.e. a suitably ordered hypercube, sometimes called indirect binary cube) when $S=\{0,1\}$, or an $s$-ary butterfly network 
for an arbitrary finite set $S$ with $\left|S\right|=s$.  
Butterfly networks are a classical tool in network theory, and we mention that butterfly-type structures also appear naturally in recursive computation,
for example in the implementation of the well-known FFT algorithm \cite{CoTu65}, 
see \cite{Lei92}.

In the boolean case, the existence of an in situ program with
$2n-1$ assignments for a bijective mapping is equivalent to the well
known~\cite{Be64} rearrangeability of the Bene\v s network (\emph{i.e.}
of the concatenation of two reversed butterflies), that is: routing a Bene\v s network can
perform any permutation of the input vertices to the output vertices.
Such a rearrangeability result can be extended to an arbitrary finite set $S$ and an $s$-ary Bene\v s network by means of the rearrangeability properties of the Clos network \cite{Hw04}. 

%

The problem of routing a general arbitrary mapping
instead of a permutation, where several inputs may have the same output,
is equivalent, up to reversing the direction of the network,
to the rearrangeable multicast routing problem:
one input may have several outputs, each output must be reachable from the associated input, and the different trees joining inputs to their outputs must be edge-disjoint.
This general problem is a classical one in network theory, where sometimes \emph{rearrangeable} is called \emph{rearrangeable non-blocking}, and a huge number of routing methods have been developed for various networks, whose aims are to minimize the number of connections and to maximize the flexibility of the routings. 
As an instance of network derived from the butterfly network, an efficient construction consists in stacking butterfly networks \cite{Le90}.
Other examples and further references can be found for instance in \cite{Hw04}.

In this paper, we are interested in networks obtained by concatenation of butterfly networks (a construction sometimes called cascading). 
A rearrangeable multicast routing method for such (boolean) networks was proposed 
in \cite{Of65}, involving five copies of the butterfly network, with possible reversions.
It was noticed in \cite{Th78} that one can remove one of these copies preserving the same rearrangeability property, yielding four copies only.
In \cite{LiCh99}, a similar construction has been given, based on (boolean) baseline networks instead of butterfly networks (yielding an equivalent result since those two $log_2N$ networks are known to be topologically equivalent, see
\cite{BeFoJM88} for instance).

In Section \ref{sec:general}, we investigate this problem under the setting of in situ programs, and we provide similar results than those cited above, 
with slight variants and complementary results (arbitrary finite sets, inversion of bijections...).
Also, this provides a practical framework that unifies those network results from the literature. 
This connection is not always clear from the way those references were written, and we feel that this survey work is interesting on its own. 
And this framework will serve again for the next section.
This yields finally an in situ program of length $4n-3$ for a general mapping.
The common general idea of those constructions involving four butterfly copies is the following: first group the vectors having same image, using two copies that provide a (unicast) rearrangeable network; then use one copy to give all those vectors a common image; and lastly use one copy to bring those images to the final required output.
The efficiency of the two last steps relies on the capability of the butterfly network to map arbitrary inputs onto consecutive outputs in same order
 (a sorting property called \emph{infra-concentrator} property in \cite{Hw04}, or \emph{packing problem} property in \cite{Lei92}).
So, the limitation of this type of multicast routing strategy is that the groups formed at the middle stage have to be exactly in the same order than the final images, thus this middle stage is (almost) totally determined, yielding a poor flexibility. 

In Section \ref{sec:boolean},  we provide a new and more sophisticated construction, relying on the same framework. It involves  arithmetical properties because of which we assume that $|S|$ is a power of 2. We mention that those propeties, and so the whole construction, can be extended to arbitrary  $|S|$ as noticed in \cite{Ga1}.
It yields new results on butterfly routing properties refining its classical sorting property (Proposition \ref{prop:compose}), and a more flexible multicast routing method for general mappings, with the same number of stages (four copies of the butterfly network), i.e. same in situ program length (Theorem \ref{th:4n}).
The improvement is to allow  a huge number of possible orderings of the groups at the middle stage (see Remark \ref{flexible} for details).

%

\bigskip

Let us now give some details from the algorithmic viewpoint.
Building assignments whose number
is linear in $n$ to perform a mapping of $S^n$ to itself is satisfying in the
following sense.  If the input data is an arbitrary mapping
$E : S^n\rightarrow S^n$ with $|S|=s$, given as a table of $n\times s^n$ values, then the
output data is a linear number of mappings $S^n\rightarrow S$ whose total
size is a constant times the size of the input data.  This means
that the in situ program of $E$ has the same size as the definition of
$E$ by its components, up to a multiplicative constant.  This
complexity bound is essentially of theoretical interest, since in terms
of effective technological applications, it may be difficult to deal with
tables of $n\times s^n$ values for large $n$.
%
Alternatively, assignments can be defined with algebraic expressions,
for instance mappings $\{0,1\}^n\rightarrow \{0,1\}$ are exactly multivariate polynomials of degree at most $n$ on $n$ variables on the binary field $\{0,1\}$.
Hence, it is interesting to deal with an input data given by
algebraic expressions of restricted size,
like polynomials of bounded degree for instance, and compare the
complexity 
of the
assignments in the output data with the input one,
for instance using polynomial assignments of a related bounded degree.
This general question
(also related to the number of gates in a chip design)
can motivate further research 
(examples are given in \cite{BuGi08}, see also Open problems \ref{open:alg} at the end of the paper).

Here, in Section \ref{sec:linear}, we prove that, in the linear case,
i.e. if the input is given by polynomials with degree at most 1,
with respect to any suitable algebraic structure for $S$ (e.g. any field, or \ $\bZ/s\bZ$), then the
assignments are in number $2n-1$ and overall are also linear.
Hence, we still obtain a program whose size is
proportional to the restricted size of the input mapping.
We mention that this decomposition method takes $O(n^3)$ steps to build the program,
and that if the mapping is invertible, then we get naturally a program for the inverse.
This result generalizes to a large extent the result in \cite{BuMo04bis} obtained for linear mappings on the binary field.
In terms of multistage interconnection networks, a similar result is given in \cite{RaBo91}, also for the binary field only. Here, we get rearrangeable non-blocking multicast routing methods for the $s$-ary Ben\v es network as soon as the input/outputs are related through a linear mapping on any suitable more general algebraic structure.
Let us insist on the fact that this result is way more general than its restriction to the boolean field.
First, there is a generalization from the boolean field $\bZ/2\bZ$ to any field, such as $\bZ/p\bZ$ for $p$ prime.
Secondly, there is a generalization to general rings
such as $\bZ/n\bZ$ for any integer $n$ (which are not necessarily fields, i.e. elements are not necessarily invertible). 
Linear mappings on such general rings are fundamental and much used in mathematics and computer science (e.g. in cryptography).
Those two generalizations should be considered as non-trivial theoretical jumps.
Also,  from the algebraic viewpoint,  this provides a new result on matrix decompositions.
\bigskip

Finally, let us mention that some of the original motivation for this research was in terms of technological applications. A permanent challenge in computer science 
consists in increasing the
performances of computations and the speed of processors.
A computer decomposes a
computation in elementary operations on elementary objects.  For
instance, a $64$ bits processor can only perform operations on $64$
bits, and any transformation of a data structure
must be decomposed in successive
operations on $64$ bits.
Then, as shown in the above example
on the exchange of the contents of two registers $x_1$ and $x_2$, 
the usual solution to ensure the completeness of the computation 
is to make copies from the initial data. 
But this solution can
generate some memory errors when the structures are too large, or at
least decrease the performances of the computations.  
Indeed, such
operations involving several registers in a micro-processor, through
either a compiler or an electronic circuit, will have either to make copies of
some registers in the cache memory or in RAM, with a loss of speed,
or to duplicate signals in the chip design itself, with an extra
power consumption.
%
On the contrary, the theoretical solution provided by in situ computation
would possibly avoid the technological problems alluded to, and hence increase the performance.

\section{In situ programs}
\label{sec:insitu}


For the ease of the exposition, we fix for the whole paper a finite set $S$ of cardinality $s=\gl S\gl$, a strictly positive integer $n$ and a mapping
$E:S^n\rightarrow S^n$. 
This paper strongly relies on the following definition.
    

\begin{definition} 
\label{def:in situ program}
{\rm
An {\it in situ program $\Pi$} of a mapping $E:S^n\rightarrow S^n$ is a finite sequence 
$$(\f_1,i_1),\ (\f_2,i_2),\ ...,\ (\f_m,i_m)$$
of {\it assignments} such that:

\noindent - for $k=1,2,...,m$,
we have $\f_k:S^n\rightarrow S$ and $i_k\in\{1,...,n\}$;

\noindent - every transformation $X=(x_1,..., x_n)\mapsto E(X)$ is computed
through the sequence
$$X=X_0, X_1,\ldots, X_{m-1}, X_m=E(X)$$
where, for $k=1,2,...,m$,
the vector $X_k$ has the same components as $X_{k-1}$ except component $x_{i_k}$ which is equal to $\f_k(X_{k-1})$.
%

In other words, $\f_k$ modifies only 
the $i_k$-th component~of~the current vector, that is: 
every assignment $(\f_k,i_k)$ of an in situ program performs the elementary operation $$x_{i_k}:=\f_k(x_1,...,x_n).$$
The {\it length} of $\Pi$ is the number $m$. The {\it signature} of $\Pi$ is the sequence $$i_1,i_2,...,i_m.$$ 
}
\end{definition}

All in situ programs considered throughout this paper operate on
consecutive components, traversing the list of all indices, possibly several
times in forward or backward order. Thus program signatures will all be
of type: $$1, 2, ..., n-1, n, n-1,...,2,1,2,...n-1,n,...$$
For ease of
exposition, 
the mappings $S^n\rightarrow S$ in the corresponding sequence of assignments
will be simply distinguished by different letters, e.g. $f_i$ denotes the
mapping affecting the variable $x_i$ on the first traversal, $g_i$ the
one affecting $x_i$ on the second traversal, and so on,
providing 
a sequence of assignment mappings
denoted 

\centerline{$f_1, f_2, ...,f_{n-1}, f_n, g_{n-1},..., g_2,g_1,...$}
\noindent where each index gives the index of the component modified by the mapping.
%
%
%
%
%
%
%
%
For instance, a
program $f_1,f_2,g_1$ on $S^2$ represents the sequence of operations:
{$x_1:=f_1(x_1,x_2);$ $x_2:=f_2(x_1,x_2);$ $x_1:=g_1(x_1,x_2).$}
\bigskip


As an example, it is easy to see that a mapping consisting in a cyclic permutation of $k$ variables in a group $S$ can be computed in $k+1$ steps using the $k$ variables only.
This is an extension of the case of the exchange of two variables.
Precisely $(x_1,\dots,x_k)\mapsto (x_2,\dots, x_k,x_1)$ is computed by the in situ program:
 $x_1:=x_1+x_2+\dots+x_k$;  $x_k:=x_1-x_2-\dots-x_k$; $\ldots$;  $x_2:=x_1-x_2-\dots-x_k$.
 This length turns out to be a minimal bound for this type of mapping, as shown in the next proposition, which we state as an example of an in situ computation property  whose proof is not so trivial.
Let us mention that this proposition has been suggested by \cite{De11}, and that \cite{Ga1} provides a similar but slightly more general result authorizing overwriting of variables.

\begin{proposition}
\label{permut}
If $S$ is a finite set, and the mapping $E : S^n\rightarrow S^n$ consists in a permutation of the $n$ variables, then an in situ program for $E$ has a length greater than $n - f +c$, where $c$ is the number of cycles of $E$ non-reduced to one element, and 
$f$ is the number of invariant elements of $E$.
\end{proposition}

\begin{proof}
%
%
In what follows, we assume that $f=0$, then the proof can be extended directly to the case where $f>0$ by applying it to the restriction of $E$ to $S^{n-f}$.
Assume there exists an in situ program $\Pi$ for $E$ of length strictly smaller than $n+c$.

Then there exists a cycle (non-reduced to a single variable) whose variables are modified once and only once each in the program (since each non-invariant variable is modified at least once).
Assume the sequence of assignments modifying the variables of that cycle transforms the variables $x_1,\dots, x_k$ into $x_2, \dots, x_k, x_1$ respectively.
Let us consider the first variable modified by the program amongst those variables, say it is $x_1$. Then it is necessarily modified by the assignment $x_1:=x_2$.

Let us consider the program $\Pi'$ formed by all assignments of the program $\Pi$ from the first one to the assignment $x_1:=x_2$, included.
The variables  which are modified by those assignments are $x_1$ and some variables $y_1,\dots y_i$.
Let us consider every other variable from $\Pi$ as a constant for $\Pi'$.

The vector $(y_1,...,y_i,x_1)$ can have at the beginning every possible value in $S^{i+1}$.
Since the permutation of variables is a bijection, this program $\Pi'$ has to compute a bijection from 
$S^{i+1}$ into $S^{i+1}$.
But this is impossible since the image of the mapping computed by $\Pi'$ has a size bounded by $|S| ^i $, because the assignment $x_1:=x_2$  ends the program $\Pi'$ where $x_2$ is a constant.
Note that  we use that $S$ is finite.
\end{proof}

\section{Multistage interconnection networks}
\label{sec:MIN}

Among formalism and terminology variants in the network theory field,
we will remain close to that of \cite{Hw04} and \cite{Lei92}. 
Also, we prefer to define a network as a directed graph, whose routing consists in choosing edges to define directed paths, rather than considering vertices as switches with several routing positions to choose. Those two formal options are obviously equivalent.
\medskip

A {\it multistage interconnection network}, or {\it MIN} for short, is a directed graph
whose set of vertices is a finite number of copies $S_1^n, S_2^n, \ldots, S_k^n$ of $S^n$, called {\it stages},
 and whose edges join elements of $S_i^n$ towards some elements of $S_{i+1}^n$ for $1\leq i<k$.
Then {\it routing} a MIN is specifying one outgoing edge from each vertex of $S_{i}^n$ for $1\leq i<k$.
A mapping $E$ of $S^n$ is {\it performed} by a routing of a MIN if
for each element $X\in S_1^n$ there is a directed path using specified edges
from $X$ to $E(X)\in S_k^n$.
The fact that a MIN performs a mapping $E$ can be seen as the reverse of a multicast communication pattern where one input may lead to several outputs. So, a MIN is called \emph{rearrangeable non-blocking multicast}
if every mapping of $S^n$ can be performed by this MIN.
The {\it concatenation} of two MINs $M,M'$ is the MIN $M\gl  M'$ obtained by identifying the last stage of $M$ and the first stage of $M'$.
\medskip

The {\it assignment network} $A_i$ is the MIN with two stages
whose edges join $(x_1,\ldots,x_n)$ to $(x_1,\ldots,x_{i-1},$ $ e, x_{i+1},\ldots, x_n)$ for an arbitrary $e\in S$. Hence each vertex has degree $s=\gl S\gl$.
With notations of Definition \ref{def:in situ program}, given an assignment $(\f_k,i_k)$ in an in situ program,
we naturally define a routing of $A_{i_k}$
by specifying the edge
between $X=(x_1,\ldots, x_n)$ and $(x_1,...,x_{i_k-1},\f_k(X),x_{i_k+1},$ $...,x_n)$.
%
\medskip

\begin{figure}[htbh] 
\centerline{\includegraphics[scale=1.6]{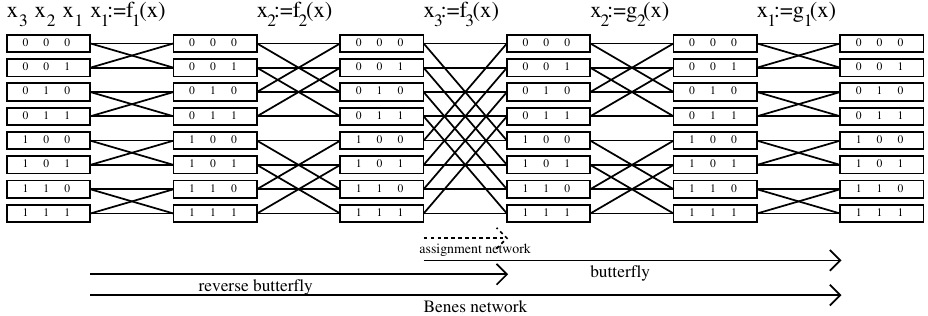}}
\centerline{ \small{ \bf Figure 1.}}
\end{figure}

The $s$-ary \emph{butterfly network}, or simply \emph{butterfly}, denoted $\Bsn$, or $\B$ for short, 
is the MIN $A_n\gl\dots $ $\gl A_2\gl A_1$.
Then $\R$ is the MIN $A_1\gl A_2\gl\dots \gl A_n$.
The usual $2$-ary {\it butterfly}, also called indirect binary cube, stands here as $B_{2,n}$.
The {\it Bene\v s network} is the network obtained from 
$\R\gl \B$ 
by replacing the two consecutive assignment networks $A_n$ by a single one. 
Note that this last
reduction is not part of the usual definition, however it is more
convenient here since two successive assignments on a same component can always
be replaced with a single one.
Note also that the historical definition
of a Bene\v s network \cite{Be64} is not in terms of butterflies, but
that ours is topologically equivalent thanks to classical results (see
\cite{Ag83} and \cite{BeFoJM88} for instance) implying that they are
equivalent in terms of mappings performed.
\medskip

The point is that the signature of an in situ program defines a MIN, and the set of assignments that realize a given mapping define the routing of the MIN by specifying some routing edges between stages.
From the above definitions, an in situ program of signature 
$n,\dots, 1$, or $1,\dots,n$, or $1,\dots, n\dots, 1$
corresponds to a routing in $\B$, or $\R$, or the Bene\v s network, respectively.
Figure 1 
gives an example for the Bene\v s network,
with corresponding in situ program $f_1,f_2,f_3,g_2,g_1$ (where indices show the modified components).
As explained above, routing this network is exactly specifying these mappings. 

\section{A formalization and survey of results for both viewpoints}
\label{sec:general}


In this section, we provide reformulations, variations, complements, or extensions for known network theory results, in terms of in situ programs. We point out that deriving these constructions from existing literature is not straightforward and is interesting on its own, notably because of various approaches and formalisms used. 
Useful references are recalled. 
The formalism and preliminary constructions introduced here will also serve as a base for the next section.

The classical property of the Bene\v s network is that it is {\it rearrangeable} (see \cite{Be64}\cite{Hw04}),
that is: for any permutation of $S^n$, there exists a routing performing
the permutation
(note that a routing performs a permutation when it defines edge-disjoint directed paths).
Theorem \ref{th:bijective} below reformulates this result.
A short proof in terms of in situ programs and using graph colouration is given in \cite{BuGiTh09}, yielding a construction of the in situ program 
in $DTIME(t.log(t))$,  where $t=n.2^n$ is the size of the table defining the mapping.
%

\begin{theorem} \label{th:bijective} 
Let $E$ be a bijective mapping on $S^n$. There exists an in situ program for $E$ of length $2n-1$ and signature
$1\dots n\dots 1$. Equivalently, 
$\R\gl  \B$ has a routing performing $E$.


\end{theorem}

From a routing for a bijection $E$, one immediately gets a routing for $E^{-1}$ by reversing the network. Hence, the mappings corresponding to the reversed arcs in the reserved network define an in situ program of $E^{-1}$. In the boolean case, we even obtain more: one just has to use exactly the same assignments but in the reserved way, as stated in Corollary \ref{cor:reverse}. 

\begin{corollary} \label{cor:reverse}
If $\Pi$ is an in situ program of a bijection $E$ on $\{0,1\}^n$, then
the reversed sequence of assignments is an in situ program of the inverse bijection $E^{-1}$.
\end{corollary}


\begin{proof} 
First, we show that operations in the program $\Pi$ are necessarily
of the form $$x_i:=x_i+h(x_1,..,x_{i-1},x_{i+1},...,x_n).$$
One can assume without loss of
generality that $i=1$. Let $x_1:=f(x_1,...,x_n)$ be an operation of $\Pi$.
Denote $$h(x_2,\dots,x_n)=f(0,x_2,\dots,x_n).$$
We necessarily have
$f(1,x_2,\dots,x_n)=1+h(x_2,\dots,x_n)$. Otherwise two different vectors
would map to the same image. This yields
$f(x_1,\dots,x_n)=x_1+h(x_2,\dots,x_n)$.
As a straightforward consequence, performing the operations in reverse order will
compute the inverse bijection $E^{-1}$.
\end{proof}
\bigskip


Now, in order to build a program for a general mapping $E$ on $S^n$, for which different
vectors may have same images, we will use a special kind
of mappings on $S^n$,
that can be computed with $n$ assignments.


%

\begin{definition}
\label{def:indexed_mapping}
{\rm
It is assumed that $S = \{0,1,\ldots,s-1\}$.
Denote $[s^n]$ the interval of integers $[0,\ldots,s^n-1]$.
The {\it index} of a vector $(x_1,x_2,\dots,x_n)$ is the integer $x_1+s.x_2+\dots+s^{n-1}.x_n$ of $[s^n]$.
For every $i\in[s^n]$, denote by $X_i$ the vector of index $i$.
The {\it distance} of two vectors $X_a,X_b$ is the integer $\Delta(X_a,X_b)=|b-a|$.

A mapping $I$ on $S^n$ is {\it distance-compatible} if for every $x,y\in S^n$, we have $\Delta(I(x),I(y))\le\Delta(x,y)$,
which is equivalent to $\Delta(I(X_a),I(X_{a+1}))\le 1$ for every $a$
 with $0\le a< s^n-1$.
}
\end{definition}

%
%
%
%

Proposition \ref{prop:direct} and Corollary \ref{cor:direct}
below provide an extension of a well-known property of the butterfly network in terms of in situ programs:
it can be used to map the first $k$ consecutive inputs onto any set of $k$ outputs in the same order.
In \cite{Of65}, this property is used in a similar way than ours,
as recalled in \cite{Hw04} (see notably Theorem 4.3.7, where this network is shown to be a \emph{multicast infra-concentrator},
and see also  \cite{Lei92}, Section 3.4.3., where this property is used to solve the \emph{packing routing problem}, with a proof similar to ours).

\begin{proposition}\label{prop:direct}
Every distance-compatible mapping $I$ on $S^n$ is computed by an in situ program with signature $1,\dots,n.$
This program $p_1, p_2,\ldots, p_n$ satisfies,
for $I(x_1,\dots,x_n)=(y_1,\ldots,y_n)$ and for each $i=1,2,\dots,n$: $$p_i(y_1,\ldots,y_{i-1},x_i,\ldots,x_n)=y_i.$$
\end{proposition}

\begin{proof}
Since each component is modified exactly one time in a program with signature $1,\dots, n$, necessarily each function $p_i$ must give its correct final value to each component $x_i$.
It remains to prove that this unique possible method is correct, 
that is the mappings $p_i$ are well defined by the above formula,
that is, for each $p_i$, a same vector cannot have two different images 
according to the definition.
Note that the given definition is partial, but sufficient for computing the image of any $x$.

Assume that $p_1,..., p_i$ are well defined.
Assume that, after step $i$,
two different vectors $x,x'$ are given the same image by the process whereas their final expected images $I(x)$ and $I(x')$ were different.
The components $x_j$, $j>i$, of $x$ and $x'$ have not been modified yet. Hence, they are equal and we deduce $\Delta(x,x')<s^i$.
On the other hand, the components $y_j$, $j\leq i$,  of $I(x)$ and $I(x')$ are equal but $I(x)\not=I(x')$.
Hence $\Delta(I(x),I(x'))\ge s^i$: a contradiction.
So $p_{i+1}$ is also well defined by the given formula.
\end{proof}

\begin{corollary}
\label{cor:direct}
Let $I$ be a mapping on $S^n$ preserving the strict ordering of a set of  consecutive vectors $X_i,\dots,X_j$, for $0\leq i<j<s^n-1$.
Then the restriction of $I$ to the set of vectors $\{X_i,\dots,X_j\}$ can be computed by an in situ program with signature $n,\dots,1.$
\end{corollary}

\begin{proof}
By assumption, the restriction of $I$ to $\{X_i,\dots,X_j\}$ is injective. 
Let $I^{-1}$ be a mapping  of $S^n$  whose restriction to $I(\{X_i,\dots,X_j\})$ is the inverse of $I$, and completed so that 
$I^{-1}$ is distance-compatible.
Applying Proposition \ref{prop:direct} to $I^{-1}$
provides a sequence of assignments performing $I^{-1}$ with signature $1,\dots, n$.
Since $I$ is injective on $\{X_i,\dots,X_j\}$, those assignments can be reversed to provide a sequence of assignments computing the restriction of $I$ to this set of vectors.
Observe that this result can be seen more simply in terms of networks:
the in situ program of the mapping $I^{-1}$ corresponds to a routing of the reversed butterfly, reversing this routing provides directly a routing of the butterfly performing the required restriction of $I$.
\end{proof}

\begin{example}
\label{ex:dist-compat}
{\rm
%
For $S=\{0,1\}$ and $n=3$,  consider the mapping $I$ defined by 
$I(X_0)=I(X_1)=X_0$, $I(X_2)=X_1$, $I(X_3)=I(X_4)=I(X_5)=X_2$ and $I(X_6)=I(X_7)=X_3$, as shown on the following left tables.
%
The mapping $I$ is computed by the in situ program $p_1,p_2,p_3$ as defined in Proposition \ref{prop:direct} and as illustrated in the following right tables.
For consistency with Figure 1, and for better readability of the index of a vector, the vector $(x_1,...,x_n)$ is written in reversed order in columns of the tables: from $x_3$ to $x_1$.

\smallskip
{\ptirm
\hfil
\begin{tabular}{|ccc|}
\hline
 $\scriptstyle x_3$ & $\scriptstyle x_2$ & $\scriptstyle x_1$ \cr
\hline
 0 & 0 & 0 \cr
 0 & 0 & 1 \cr
 0 & 1 & 0 \cr
 0 & 1 & 1 \cr
 1 & 0 & 0 \cr
 1 & 0 & 1 \cr
 1 & 1 & 0 \cr
 1 & 1 & 1 \cr
\hline
\end{tabular}
$I\atop \rightarrow$
\begin{tabular}{|ccc|}
\hline
 $\scriptstyle y_3$ & $\scriptstyle y_2$ & $\scriptstyle y_1$ \cr
\hline
 0 & 0 & 0 \cr
 0 & 0 & 0 \cr
 0 & 0 & 1 \cr
 0 & 1 & 0 \cr
 0 & 1 & 0 \cr
 0 & 1 & 0 \cr
 0 & 1 & 1 \cr
 0 & 1 & 1 \cr
\hline
\end{tabular}
\hskip 1.2cm
%
\begin{tabular}{|ccc|}
\hline
 $\scriptstyle x_3$ & $\scriptstyle x_2$ & $\scriptstyle x_1$ \cr
\hline
 0 & 0 & 0 \cr
 0 & 0 & 1 \cr
 0 & 1 & 0 \cr
 0 & 1 & 1 \cr
 1 & 0 & 0 \cr
 1 & 0 & 1 \cr
 1 & 1 & 0 \cr
 1 & 1 & 1 \cr
\hline
\end{tabular}
$p_1\atop \rightarrow$
\begin{tabular}{|ccc|}
\hline
 $\scriptstyle x_3$ & $\scriptstyle x_2$ & $\scriptstyle y_1$ \cr
\hline
 0 & 0 & 0 \cr
 0 & 0 & 0 \cr
 0 & 1 & 1 \cr
 0 & 1 & 0 \cr
 1 & 0 & 0 \cr
 1 & 0 & 0 \cr
 1 & 1 & 1 \cr
 1 & 1 & 1 \cr
\hline
\end{tabular}
$p_2\atop \rightarrow$
\begin{tabular}{|ccc|}
\hline
 $\scriptstyle x_3$ & $\scriptstyle y_2$ & $\scriptstyle y_1$ \cr
\hline
 0 & 0 & 0 \cr 
 0 & 0 & 0 \cr
 0 & 0 & 1 \cr
 0 & 1 & 0 \cr
 1 & 1 & 0 \cr
 1 & 1 & 0 \cr
 1 & 1 & 1 \cr
 1 & 1 & 1 \cr
\hline
\end{tabular}
$p_3\atop \rightarrow$
\begin{tabular}{|ccc|}
\hline
 $\scriptstyle y_3$ & $\scriptstyle y_2$ & $\scriptstyle y_1$ \cr
\hline
 0 & 0 & 0 \cr
 0 & 0 & 0 \cr
 0 & 0 & 1 \cr
 0 & 1 & 0 \cr
 0 & 1 & 0 \cr
 0 & 1 & 0 \cr
 0 & 1 & 1 \cr
 0 & 1 & 1 \cr
\hline
\end{tabular}
\hfil
}
\bigskip
}
\end{example}

\begin{definition}
\label{def:part-seq}
{\rm
We call {\it partition-sequence of $S^n$}
 a sequence $$P=(P_0,P_1,\dots,P_k)$$  of subsets of
$S^n$, for some integer $k\geq 0$,  such that the non-empty subsets in the sequence form a partition of $S^n$.
Then, we denote by $I_P$ the mapping on $S^n$ which maps 
 $X_0,\dots,X_{s^n-1}$ respectively to
 
$$
\overbrace{X_0,\dots, X_0,}^{|P_0|}
\overbrace{X_1,\dots, X_1,}^{|P_1|}
\dots,
\overbrace{X_k,\dots, X_k}^{|P_k|}.$$
Observe that $I_P$  is well defined since the sum of sizes of the subsets equals $s^n$,
and that $I_P$ depends only on the sizes of the subsets and their ordering.
Observe also that if no subset is empty, then $I_P$ is distance-compatible since, by construction, $\Delta(I(X_{a}),I(X_{a+1}))\le 1$ for every $a$.
}
\end{definition}

\begin{example}
\label{ex:part-seq}
{\rm
The mapping $I$ from Example \ref{ex:dist-compat} equals $I_{\tilde P}$ for the partition-sequence $\tilde P=(\tilde P_0,\tilde P_1,\tilde P_2,\tilde P_3)$ of $\{0,1\}^3$ such that $[\gl \tilde P_0\gl,\gl \tilde P_1\gl,\gl \tilde P_2\gl,\gl \tilde P_3\gl]=[2,1,3,2]$.
}
\end{example}

\begin{definition}
\label{def:decomp}
{\rm
Let $E$ be a mapping on $S^n$,
and $P=(P_0,P_1,\dots,P_k)$ be a partition-sequence of $S^n$
whose underlying partition of $S^n$ is
given by the inverse images of $E$, that is precisely:
for every $0\leq i\leq k$, if $P_i\not=\emptyset$ then  there exists (a unique) $y_i\in S^n$ such that $P_i=E^{-1}(y_i)$.
Then, we call \emph{\trio of $E$} a triple of mappings $(F,I,G)$ on $S^n$ such that:

\begin{itemize}
\item $I$ is the mapping $I_P$;
\item $G$ is bijective and maps the set $P_i$ onto the set $I^{-1}(X_i)$, for every $0\leq i\leq k$\break (it is arbitrary within each set $P_i$);
\item $F$ is bijective and maps $X_i$ to $y_i$, for every $0\leq i\leq k$ such that  $P_i\not=\emptyset$\break (it is arbitrary for other values $X_i$). 
\end{itemize}
By construction, we have $$E=F\circ I\circ G.$$
}
\end{definition}


Using this construction with no empty subset in the sequence $P$, we obtain Theorem \ref{cor:5n} below, which significantly improves the result of \cite{BuMo04} where boolean mappings on $\{0,1\}^n$ are computed in $n^2$ steps. This result is similar, in terms of in situ programs, to the result of \cite{Of65} for boolean mappings, as presented in~\cite{Hw04}.


\begin{theorem} \label{cor:5n}
For every finite set $S$, every mapping $E$ on $S^n$
can be computed by an in situ program of signature
$1\dots n\dots 1\dots n\dots 1\dots n$ and length $5n-4$ the following way: 
 
 \begin{itemize}

\item Consider any \trio $(F,I,G)$ of $E$ with no empty subset in the sequence $P$

\item Use Theorem~\ref{th:bijective} to compute $G$ (resp. $F$) by a program of signature $1\dots n\dots 1$ (resp. $n\dots 1\dots n$).

\item  Use Proposition~\ref{prop:direct}, to compute $I$  by a program of signature $1\dots n$.

\item Reduce into one assignment the consecutive assignments operating on the same component.
\end{itemize}

In terms of MIN, we get a routing of 
$\R\gl\B\gl\R\gl\B\gl\R$  performing $E$,
and a multicast routing of $\B\gl\R\gl\B\gl\R\gl\B$ performing $E^{-1}$.
\end{theorem}

\begin{proof}
Consider any \trio $(F,I,G)$ of $E$ with no empty subset in the sequence $P$. Then the mapping $I$ is distance compatible, as already observed.
By Theorem~\ref{th:bijective}, $G$ (resp. $F$) can be computed by a program of signature $1\dots n\dots 1$ (resp. $n\dots 1\dots n$).
By Proposition~\ref{prop:direct}, $I$ is computed by a program of signature $1\dots n$. By composition and by reducing two successive assignments of the same variable in one, $E$ is computed by a sequence of $5n-4$ assignments of signature $1\dots n\dots 1\dots n\dots 1\dots n$.
\end{proof}
\bigskip

Now we can refine Theorem \ref{cor:5n} to get Theorem \ref{cor:4ngene} below.
This modification is similar to that noticed in \cite{Th78}
about \cite{Of65}, 
and is similar to the construction of \cite{LiCh99} in terms of boolean baseline networks.
Observe that this refinement is an improvement in terms of number of assignments (number of routing edges in the MIN setting), but not in terms of flexibility, since the ordering of subsets in the  sequence $P$ of the \trio of $E$ has now to be the same than the ordering of the outputs of $E$.

\begin{theorem} \label{cor:4ngene}
For every finite set $S$, every mapping $E$ on $S^n$
can be computed by an in situ program of signature
$1\dots n\dots 1\dots n\dots 1$ and length $4n-3$ the following way: 
 
 \begin{itemize}

\item Consider a \trio $(F,I,G)$ of $E$ with no empty subset in the sequence $P=(P_0,...,P_k)$ and such that $P_i=E^{-1}(y_i)$ where $y_0,...,y_k$ are the images of $E$ in increasing ordering

\item Use Theorem~\ref{th:bijective} to compute $G$ by a program of signature $1\dots n\dots 1$.

\item Use Proposition~\ref{prop:direct}, to compute $I$  by a program of signature $1\dots n$.

\item Use Corollary~\ref{cor:direct}, to compute the restriction of $F$ to the image of $I\circ G$ by a program of signature $n\dots 1$.

\item Reduce into one assignment the consecutive assignments operating on the same component.

\end{itemize}

In terms of MIN, we get a routing of 
$\R\gl\B\gl\R\gl\B$  performing $E$,
and a multicast routing of $\B\gl\R\gl\B\gl\R$ performing $E^{-1}$.
\end{theorem}

\begin{proof}
The proof is similar to that of Theorem \ref{cor:5n} except for the bijection $F$.
Here, by the choice of $P$, 
the restriction of $F$ to the image of $I\circ G$ maps consecutive vectors onto a set of vectors preserving the ordering. Hence
it satisfies the hypothesis of Corollary \ref{cor:direct}.
\end{proof}

\begin{figure}[htbh]
\centerline{\includegraphics[scale=1.6]{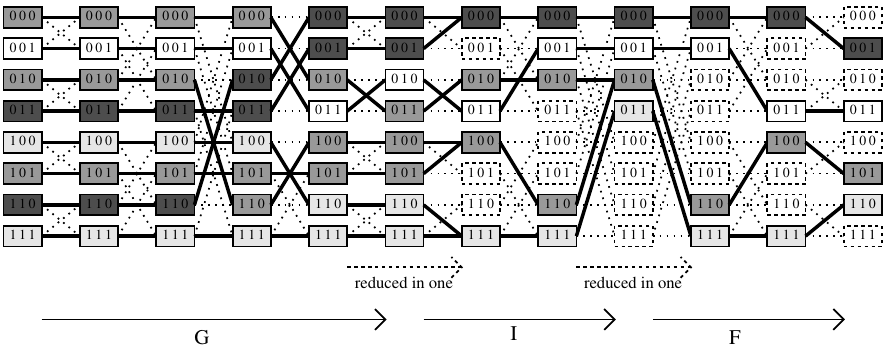}}
\centerline{ \small{ \bf Figure 2.}}
\end{figure}

\begin{example}
{\rm
Figure 2 
gives an example for the construction of Theorem \ref{cor:4ngene}.
The elements of $\{0,1\}^3$ are grouped by the bijection $G$ at stage 6, in the same ordering than the images of $E$.
Then, at stage 9 all elements with same final image have been given a same image by $I$, again in the same ordering than the images of $E$. Hence, at last, the restriction of bijection $F$ can finalize the mapping $E$ in 3 stages only.
}
\end{example}

\begin{remark}
\label{rk:register}
{\rm
To end this section, let us notice that, due to the fact that successive assignments operate on consecutive components,
successive assignments of type $S^{mn}\rightarrow S$
can be grouped in assignments of
fewer variables on a larger base set $S^m$
defining 
successive mappings  $S^{mn}\rightarrow S^{m}$:
$$\underbrace{f_{nm},\ldots,f_{n.(m-1)+1}}_{\tilde f_n} ,\ldots, \underbrace{f_m,\ldots f_2,f_1,g_2,\ldots, g_m}_{\tilde f_1}, \ldots, \underbrace{g_{n.(m-1)+1},\ldots,g_{nm}}_{\tilde g_n}. $$
Hence, for instance, the case $S=\{0,1\}^m$ can be reduced to the case $S=\{0,1\}$. This is a particular case of the register integrability property
described~in~\cite{BuGi08}: the signatures of type $1,\dots,n,\dots,1,\dots$ allow to adapt in situ programs to memory registers with non-constant sizes.
}
\end{remark}

\begin{open problem}
\label{open:length}
{\rm
Up to our knowledge, no better bound than $4n-3$ is known for the case of arbitrary mappings. It would be interesting to improve this bound, and to find the best possible general bound. Experiments on a computer make us think that the factor $4$ could be replaced with a factor $3$.
Note that the in situ program may not have a signature with consecutive  indices, or, in other words, that other combinations of assignment networks than the butterfly network may be used.
}
\end{open problem}




\section{A more flexible new method}
\label{sec:boolean}


The more involved method given here 
is a refinement of the method 
 given by Theorem \ref{cor:5n}.
It is completely new with respect to network literature, and had been presented in the preliminary conference paper \cite{BuGiTh09}.
It provides the same number of assignments than Theorem \ref{cor:4ngene} and a better flexibility.


We still use a \trio $(F,I,G)$ but the sequence $P$ will possibly contain empty sets, and will be suitably ordered with respect to the
sizes of its elements, in order to satisfy some boolean arithmetic properties.
So doing, the intermediate mapping $I=I_P$ will have a so-called suffix-compatibility property.
We show that the composition of a mapping having this property
with any in situ program with signature $1,\dots, n$ can be also computed in $n$ steps.
Hence the composition of $I$ with the first $n$ steps of the in situ program of  the bijection $F$
can also be computed with $n$ assignments, performing the computation of $F\circ I$ in $2n-1$ steps instead of $3n-2$.

Each intermediate result in this section provides a new result on in situ programs with signature $1, \dots, n$, or equivalently on routing properties of the butterfly network. 
In terms of multicast routing strategy, the flexibility of this method comes from the freedom one has in building a suitable ordering for the sequence $P$, as detailed in Remark \ref{flexible}.

In the whole section, we will fix $S=\{0,1\}$.
The method is given here when $S=\{0,1\}$, hence it is directly available, by extension, when $S=\{0,1\}^m$ (cf. Remark \ref{rk:register}).
However it can be extended to a set $S$ of aribtrary size. While this paper was under publication process, it has been noticed in \cite{Ga1}
that Definition \ref{def:block-sequence} and Lemma \ref{lm:partition} below could be formulated using any integer $q=\mid S\mid$ instead of $2$, by means of an arithmetical property. Then the rest of the construction can be adapted directly.

\begin{definition} \label{def:block-sequence}
A \emph{block-sequence} $[v_0,v_1,\dots,v_{2^n-1}]$ 
is a sequence of $2^n$ non-negative integers such that,
for every $i=0\dots n$,
the sum of values in each of the consecutive \emph{blocks} of size $2^i$ is a multiple of $2^i$,
that is,  for all $0\leq j < 2^{n-i}$:
$$\sum_{j 2^i\ \leq\  l\ < \ (j+1) 2^i}v_l\ =\ 0\ \hbox{\rm mod}\ 2^i.$$
\end{definition}

\begin{lemma} \label{lm:partition}
Every sequence of $2^n$ non-negative integers whose sum equals $2^n$
can be reordered in a {block-sequence}.
%
\end{lemma}


\begin{proof}
The ordering is built inductively.
Begin at level $i=0$ with $2^n$ blocks of size 1 having each value in the sequence.
At level $i+1$,  form consecutive pairs of blocks $[B,B']$ that have values $v,v'$ of same parity and define the value of this new block to be $(v+v')/2$.
Each new level doubles the size of blocks and divides their number by 2.
The construction is valid since the sum of values of blocks at level $i$ is $2^{n-i}$.
\end{proof}



%

\begin{example}
\label{ex:blocks}
{\rm
We illustrate below the process described in the proof of Lemma \ref{lm:partition} ($n=4$ and each block has its value as an exponent):

\begin{center}
\begin{tabular}{c}
$[4]^4, [1]^1, [1]^1, [1]^1, [1]^1, [1]^1,[1]^1,[3]^3,[3]^3, [0]^0,[0]^0,[0]^0,[0]^0,[0]^0,[0]^0,[0]^0$\cr
$[4,0]^2, [1, 1]^1, [1, 1]^1, [1,1]^1,[3,3]^3,[0,0]^0,[0,0]^0,[0,0]^0$\cr
$[4,0,0,0]^1, [1,1,3,3]^2, [1, 1,1, 1]^1, [0,0,0,0]^0$\cr
$[4,0,0,0,1, 1,1, 1]^1, [1,1,3,3,0,0,0,0]^1$\cr
$[4,0,0,0,1, 1,1, 1,1,1,3,3,0,0,0,0]^1$\cr
\end{tabular}
\end{center}
}
\end{example}

\begin{remark}
\label{flexible}
{\rm
Let us anticipate the sequel of the detailed construction and already explain roughly in what respect the use of block-sequences will allow a more flexible routing strategy than the classical construction recalled in Section \ref{sec:general}.
As shown before, this construction consists in grouping and sorting
pre-images of the mapping at some middle stage of the in situ program.
The ordering of these pre-images is determined and thus the sequence of assignments
is completely constrained at this middle stage.
There is essentially \emph{one} available sorting to get the $4n-3$ length (up to a few possible shifts along the ordering).

In the more involved construction given in the present section, we will apply Lemma \ref{lm:partition} to the sequence of integers given by the cardinalities of the pre-images of the mapping to compute.
We will prove later that \emph{any} ordering of the pre-images whose cardinalities satisfy the block-sequence property can be used at this middle stage.
The point is that, given a sequence of integers, there is a number of such possible orderings given by block-sequences, built as in the proof of Lemma \ref{lm:partition}, and also a number of possible associations between the pre-images and those cardinality integers.

For instance, a given a block-sequence ordering can be   represented as a bracket system, forming a binary tree,  the following way, continuing Example \ref{ex:blocks}:

\begin{center}
\begin{tabular}{ccccccccccccccccccccccccccccccccc}

[&4&  ,&  0&  ,&  0&  ,&  0&  ,&  1&  ,&   1&  ,&  1&  ,&   1&  ,&  1&  ,&  1&  ,&  3&  ,&  3&  ,&  0&  ,&  0&  ,&  0&  ,&  0&]\\

[&4&  ,&  0&  ,&  0&  ,&  0&  ,&  1&  ,&   1&  ,&  1&  ,&   1&],[&1&  ,&  1&  ,&  3&  ,&  3&  ,&  0&  ,&  0&  ,&  0&  ,&  0&]\\

[&4&  ,&  0&  ,&  0&  ,&  0&],[&1&  ,&   1&  ,&  1&  ,&   1&],[&1&  ,&  1&  ,&  3&  ,&  3&],[&0&  ,&  0&  ,&  0&  ,&  0&]\\

[&4&  ,&  0&],[&0&  ,&  0&],[&1&  ,&   1&],[&1&  ,&   1&],[&1&  ,&  1&],[&3&  ,&  3&],[&0&  ,&  0&],[&0&  ,&  0&]\\

[&4&],[&1&],[&1&],[&1&],[&1&],[&1&],[&1&],[&3&],[&3&],[&0&],[&0&],[&0&],[&0&],[&0&],[&0&],[&0&]\\

\end{tabular}

\end{center}

Then, one can always permute \emph{any} sons of a node in the tree and still get a block-sequence.
So, from one block-sequence, one can obtain potentially a number of available block-sequences (there are $2^{2^n-1}$  permutations of the leaves obtained by this way for such a tree with $2^n$ leaves).
For instance, making such permutations in Example \ref{ex:blocks} leads to the following possible block-sequences (brackets have been added to identify blocks that have been permuted):

\hskip 2cm
\vbox{
$[4,0,0,0,1, 1,1, 1,1,1,3,3,0,0,0,0]$;

$[\ [1,1,3,3,0,0,0,0]\ ,\ [4,0,0,0,1, 1,1, 1]\ ]$;

$[\ [1, 1,1, 1]\ ,\ [4,0,0,0]\ ,1,1,3,3,0,0,0,0]$;

$[\ [1, 1,1, 1]\ ,\ [0,0,\ [0]\ ,\ [4]\ ]\ ,\ [0,0,0,0]\ ,\ [1,1,3,3]\ ]$;

etc.
}

Of course, various possible block-sequences may be obtained independently from such permutations.
Moreover, given a block-sequence, pre-images having same cardinality may be associated with \emph{any} occurrence of the corresponding integer in the sequence, leading to a number of possibilities for building the in situ program.
For instance, if every pre-image has size 2, then the sequence of integers to consider is $[2,2,\ldots,2, 2, 0, 0 , \ldots,0,0]$, which is already a block-sequence and is  invariant under permutation of the non-zero integers. In this case, \emph{any} ordering of the pre-images can be used at the middle stage
to provide  finally a $4n-3$ length program.

All those  ``\emph{any}'' in this new construction, compared with the ``\emph{one}'' in the known construction, witness how the method is more flexible.


}
\end{remark}

\begin{definition}\label{def:suffix}
{\rm
For a vector $(x_1,\dots,x_n)$,
we call {\it prefix of order $k$}, resp. {\it suffix of order $k$},
the vector $(x_1,\dots, x_k)$,
resp. $(x_k,\dots, x_n)$.
¨

A mapping $I$ of $\{0,1\}^n$ is called
\emph{suffix-compatible} if,
for every $1\leq k\leq n$, if two vectors $X,X'$ have same suffixes
of order $k$, then their images $I(X),I(X')$
also have same suffixes of order $k$.
}
\end{definition}

\begin{lemma}\label{lm:suffix}
Let $P=(P_0,P_1,\dots,P_{2^n-1})$ be a partition-sequence of
$\{0,1\}^n$ such that  $[\gl P_0\gl,\gl P_1\gl,\dots,\gl P_{2^n-1}\gl]$ is a block-sequence. Then the mapping $I_P$ on $\{0,1\}^n$
is suffix-compatible.
\end{lemma}


\begin{proof} 
The sketch of the proof is the following.
First, define the $j$-th block of level $i$ of $\{0,1\}^n$
as the set of vectors whose part has index ${j2^i\leq l< (j+1)2^i}$.
Observe that the inverse image by $I_P$ of a block
is a union of consecutive blocks of same level.
The result follows. 

Let us now detail the proof.
For $0\leq i\leq n$ and $j\in [2^{n-i}]$, define
the $j$-th block at level $i$ of $\{0,1\}^n$ as $$V_{i,j}=\{X_l: l\in[j2^i, (j+1)2^i-1]\}.$$  

(i) First, we prove that, for every $i,j$ as above, there exists
 $k,k'\in [2^{n-i}]$, such that $$I_P^{-1}(V_{i,j})=\bigcup_{k\leq l\leq k'} V_{i,l}.$$

Let us call interval of $\{0,1\}^n$ the set of vectors $X_l$ for $l$ belonging to an interval of $[2^n]$.
First, notice that
the inverse image by $I_P$ of an interval of $\{0,1\}^n$ is an interval
of $\{0,1\}^n$.
By definition of $I_P$,
we have $\left| I_P^{-1}(V_{i,j}) \right| = \sum_{j2^i\leq l< (j+1)2^i}v_l$.
Remark that $I_P^{-1}(V_{i,j})$ may be empty, when $v_l=0$ for all $l\in [j2^i,(j+1)2^i]$.
Since $[v_0,\dots,v_{2^n-1}]$ is a block sequence, we have
$\sum_{j2^i\leq l< (j+1)2^i}v_l\ =0\mod 2^i$.
Hence,  $\left| I_P^{-1}(V_{i,j}) \right| =0\mod 2^i$.

For a fixed $i$, we prove the result by induction on $j$.
If $j=0$ then $\left| I_P^{-1}(V_{i,0})\right| = k.2^i$ for some $k\in [2^{n-i}]$.
If $I_P^{-1}(V_{i,0})$ is not empty, then it is an interval of $\{0,1\}^n$ containing $(0,...,0)$
by definition of $I_P$. Since this interval has a size $k.2^i$ multiple of $2^i$,
it is of the form $\bigcup_{0\leq l\leq k} V_{i,l}$.

If the property is true for all $l$ with $0\leq l<j$,
then $I_P^{-1}\bigl(\bigcup_{0\leq l< j}V_{i,l})=\bigcup_{0\leq l\leq j'} V_{i,l}$.
Since $\left| I_P^{-1}(V_{i,j})\right| = k.2^i$ for some $k\in [2^{n-i}]$,
we must have $I_P^{-1}\bigl(\bigcup_{0\leq l\leq j}V_{i,l})=\bigcup_{0\leq l\leq j'+k} V_{i,l}$,
hence $I_P^{-1}(V_{i,j})=\bigcup_{j'< l\leq j'+k'} V_{i,l}.$
\bigskip

(ii) Now, we prove the lemma.
Assume $a=(a_1,...,a_n)$ and $b=(b_1,...,b_n)$ have same suffix of order $i$.
For all $l\geq i$ we have
$a_l=b_l$. Let $c\in S^n$ be defined by
$c_n=a_n=b_n,\ldots, c_i=a_i=b_i, c_{k-1}=0,\ldots, c_1=0$.
Let $\phi(x)$ denote the index of vector $x$.
We have $\phi(c)=0\mod 2^{i-1}$, that is $\phi(c)=j.2^{i-1}$ for some $j\in [2^{n-i+1}]$.
And $\phi(a)$ and $\phi(b)$ belong to the same interval $[j.2^{i-1}, (j+1).2^{i-1}-1]$ whose elements have same components for $l\geq i$.
That is $a$ and $b$ belong to $V_{i-1,j}$. 
By (i), the inverse images of intervals of type $V_{i-1,k}$
by $I_P$ are unions of such consecutive intervals.
Hence the image of an interval $V_{i-1,j}$ by $I_p$ is an interval contained
in an interval $V_{i-1,k}$ for some $k\in [2^{n-i+1}]$.
Hence $I_P(a)$ and $I_P(b)$ have same components $l\geq i$.
\end{proof}

%
%
%
%
%


\smallskip

\begin{example}
\label{ex:suffix-comp}
{\rm
First consider again the mapping $I_{\tilde P}$ from Examples \ref{ex:dist-compat} and \ref{ex:part-seq} obtained from the partition-sequence $\tilde P=(\tilde P_0,\tilde P_1,\tilde P_2,\tilde P_3)$
of $\{0,1\}^3$ such that $[\gl \tilde P_0\gl,\gl \tilde P_1\gl,\gl \tilde P_2\gl,\gl \tilde P_3\gl]=[2,1,3,2]$, which is not a block-sequence.
Observe that this mapping is not suffix-compatible, since
$(x_1,x_2,x_3)=(0,1,0)$ and  $(x'_1,x'_2,$ $x'_3)=(1,1,0)$ have same suffix of order 2 equal to $(1,0)$,
but
$I(x_1,x_2,x_3)=(1,0,0)$ and  $I(x'_1,x'_2,x'_3)=(0,1,0)$ have not same suffix
of order $2$.

Now consider the mapping $I_{P}$ shown on the next tables, obtained from  the partition-sequence $P=(P_0,P_1,P_2,P_3)$ such that $[\gl P_0\gl,\gl P_1\gl,\gl P_2\gl,\gl P_3\gl]=[1,3,2,2]$, which is a block-sequence. Then one can check that $I_P$ is suffix-compatible, as claimed by Lemma \ref{lm:suffix}.

\hfil
{\ptirm
\noindent
\begin{tabular}{|ccc|}
\hline
 $\scriptstyle x_3$ & $\scriptstyle x_2$ & $\scriptstyle x_1$ \cr
\hline
 0 & 0 & 0 \cr
 0 & 0 & 1 \cr
 0 & 1 & 0 \cr
 0 & 1 & 1 \cr
 1 & 0 & 0 \cr
 1 & 0 & 1 \cr
 1 & 1 & 0 \cr
 1 & 1 & 1 \cr
\hline
\end{tabular}
$I_P\atop \rightarrow$
\begin{tabular}{|ccc|}
\hline
 $\scriptstyle i_3$ & $\scriptstyle i_2$ & $\scriptstyle i_1$ \cr
\hline
 0 & 0 & 0 \cr
 0 & 0 & 1 \cr
 0 & 0 & 1 \cr
 0 & 0 & 1 \cr
 0 & 1 & 0 \cr
 0 & 1 & 0 \cr
 0 & 1 & 1 \cr
 0 & 1 & 1 \cr
\hline
\end{tabular}
}
\hfil

}
\end{example}

%
%
%

\begin{proposition} \label{prop:compose}
Let $I$ be a suffix-compatible mapping  on $\{0,1\}^n$ 
and let $B$ be a mapping  on $\{0,1\}^n$
computed by an in situ program $b_1,\ldots,b_n$.
The mapping $B\circ I$  is computed by an in situ
program with signature $1,\dots, n$, namely $p_1, p_2,\ldots, p_n$, with, for $B\circ I(x_1,\ldots, x_n)=(y_1,\ldots,y_n)$: 
$$p_i(y_1,\ldots,y_{i-1},x_i,\ldots,x_n)=y_i.$$
\end{proposition}

Observe that Proposition \ref{prop:compose} provides a new property of the butterfly network,  which refines the classical sorting property of this network recalled in Proposition \ref{prop:direct}.
We mention also that Proposition \ref{prop:compose} is stated and proved for a general mapping $B$, but we will use it in what follows only when $B$ is bijective.

\begin{proof}
Just as for Proposition \ref{prop:direct},
assume that $p_1,..., p_i$ are well defined by the necessary above formula,
and that, after step $i$,
two different vectors $x,x'$ are given the same image by the process whereas their final expected images $y=B\circ I(x)$ and $y'=B\circ I(x')$ were different (hence $I(x)\not=I(x')$).
By construction, $y,y'$ have a same prefix $P$ of order $i$ and
$x,x'$ have a same suffix $Q$ of order $i+1$. Moreover, since $I$ is suffix-compatible,
the vectors $I(x),I(x')$ also have a same suffix $R$ of order $i+1$.
Hence one has some relation $F(I(x))=F(uR)=y=Pv$ and $F(I(x'))=F(u'R)=y'=Pv'$, where $u,u'$ are some prefix, and $v,v'$ are some suffix. 
Let $F_i$ be the mapping defined by the first $i$ assignments defining $F$, that is dealing with components $1,\dots,i$. 
Since $F(uR)=Pv$, we have $F_i(uR)=PR$. And since $F(u'R)=Pv'$ we have $F_i(u'R)=PR$. Hence $F_i(uR)=F_i(u'R)$.
Since $F$ is computed by computing $F_i$ first, we get that
$F(uR)=F(u'R)$, that is $y=y'$, a contradiction with our assumption.
\end{proof}

\begin{example}
\label{ex:compose}
{\rm
Consider the bijective mapping $B$ computed by the program $b_1,b_2,$ $b_3$ shown on the next left tables (each step modifies one column). And consider  the suffix-compatible mapping $I_P$ from Example \ref{ex:suffix-comp} above.
Then, by Proposition \ref{prop:compose}, the composition $B\circ I_P$ 
is computed by the program $p_1,p_2,p_3$ as shown on the next right tables.
One can check that, at each step, two vectors 
which have same image through $1\leq k\leq 3$ assignments will have eventually
same images.
\smallskip

\noindent
{\ptirm
\hfil
\begin{tabular}{|ccc|}
\hline
 $\scriptstyle i_3$ & $\scriptstyle i_2$ & $\scriptstyle i_1$ \cr
\hline
 0 & 0 & 0 \cr
 0 & 0 & 1 \cr
 0 & 1 & 0 \cr
 0 & 1 & 1 \cr
 1 & 0 & 0 \cr
 1 & 0 & 1 \cr
 1 & 1 & 0 \cr
 1 & 1 & 1 \cr
\hline
\end{tabular}
$B\atop \rightarrow$
\begin{tabular}{|ccc|}
\hline
 $\scriptstyle b_3$ & $\scriptstyle b_2$ & $\scriptstyle b_1$ \cr
\hline
 0 & 0 & 0 \cr
 1 & 0 & 1 \cr
 0 & 1 & 1 \cr
 1 & 1 & 0 \cr
 0 & 1 & 0 \cr
 0 & 0 & 1 \cr
 1 & 0 & 0 \cr
 1 & 1 & 1 \cr
\hline
\end{tabular}
\hskip 1.2cm
\begin{tabular}{|ccc|}
\hline
 $\scriptstyle x_3$ & $\scriptstyle x_2$ & $\scriptstyle x_1$ \cr
\hline
 0 & 0 & 0 \cr
 0 & 0 & 1 \cr
 0 & 1 & 0 \cr
 0 & 1 & 1 \cr
 1 & 0 & 0 \cr
 1 & 0 & 1 \cr
 1 & 1 & 0 \cr
 1 & 1 & 1 \cr
\hline
\end{tabular}
$p_1\atop \rightarrow$
\begin{tabular}{|ccc|}
\hline
 $\scriptstyle x_3$ & $\scriptstyle x_2$ & $\scriptstyle y_1$ \cr
\hline
 0 & 0 & 0 \cr
 0 & 0 & 1 \cr
 0 & 1 & 1 \cr
 0 & 1 & 1 \cr
 1 & 0 & 1 \cr
 1 & 0 & 1 \cr
 1 & 1 & 0 \cr
 1 & 1 & 0 \cr
\hline
\end{tabular}
$p_2\atop \rightarrow$
\begin{tabular}{|ccc|}
\hline
 $\scriptstyle x_3$ & $\scriptstyle y_2$ & $\scriptstyle y_1$ \cr
\hline
 0 & 0 & 0 \cr
 0 & 0 & 1 \cr
 0 & 0 & 1 \cr
 0 & 0 & 1 \cr
 1 & 1 & 1 \cr
 1 & 1 & 1 \cr
 1 & 1 & 0 \cr
 1 & 1 & 0 \cr
\hline
\end{tabular}
$p_3\atop \rightarrow$
\begin{tabular}{|ccc|}
\hline
 $\scriptstyle y_3$ & $\scriptstyle y_2$ & $\scriptstyle y_1$ \cr
\hline
 0 & 0 & 0 \cr
 1 & 0 & 1 \cr
 1 & 0 & 1 \cr
 1 & 0 & 1 \cr
 0 & 1 & 1 \cr
 0 & 1 & 1 \cr
 1 & 1 & 0 \cr
 1 & 1 & 0 \cr
\hline
\end{tabular}
\hfil
\bigskip
}
}
\end{example}


%
%

Now, given a mapping $E$  of $S^n$,
using a 
 \trio of $E$ for a sequence $P$
 whose sequence of cardinalities is a block-sequence,
 we can improve the result of Section 3 in terms of flexibility.
Indeed, the only constraint is now to have that the sets of vectors having a same final image are grouped after the first bijection according to any block-sequence representing the sequence of cardinalities of those sets. And there is a number of such possible block-sequences (cf. construction of Lemma \ref{lm:partition}).

\begin{theorem} \label{th:4n}
Every mapping $E$ on $\{0,1\}^n$ is computed by an in situ program of length
$4n-3$ and signature $1\dots n\dots 1\dots n\dots 1$ the following way:

 \begin{itemize}

\item Consider a \trio $(F,I,G)$ of $E$ with
$P=(P_0,P_1,\dots,$ $P_{2^n-1})$ 
 such that $[\gl P_0\gl,\gl P_1\gl,\dots,\gl P_{2^n-1}\gl]$ is a block-sequence (built by Lemma \ref{lm:partition})
 
\item Use Theorem~\ref{th:bijective} to compute $G$ and $F$ by programs with signature $1\dots n\dots 1$. 

\item Call $B$ the in situ program formed by the $n$ first assignments of the program of $F$, and
use Proposition~\ref{prop:compose} to compute $B\circ I$  by a program with signature~$1,\dots, n$.

\item Reduce into one assignment the consecutive assignments operating on the same component.

\end{itemize}

In terms of MIN, we get a routing of 
$\R\gl\B\gl\R\gl\B$  performing $E$,
and a multicast routing of $\B\gl\R\gl\B\gl\R$ performing $E^{-1}$.
\end{theorem}

\begin{proof}
Let  $(F,I,G)$ be a \trio of $E$ for a sequence
 $P=(P_0,P_1,\dots,$ $P_{2^n-1})$ 
 such that $[\gl P_0\gl,\gl P_1\gl,\dots,\gl P_{2^n-1}\gl]$ is a block-sequence
 (it exists thanks to Lemma \ref{lm:partition}).
By Theorem~\ref{th:bijective}, $G$ (resp. $F$) can be computed by a program of signature $1\dots n\dots 1$ (resp. $1\dots n\dots 1$).
By Lemma \ref{lm:suffix},
the mapping $I=I_P$ on $\{0,1\}^n$
is suffix-compatible.
Call $B$ the mapping computed by the $n$ first assignments $b_1,...,b_n$ of the program of $F$.
By Proposition~\ref{prop:compose}, $B\circ I$ is also computed by a program of signature $1\dots n$. Then, by composition and by reducing two successive assignments of the same variable in one, the mapping $E=F\circ I\circ G$ is computed by a sequence of $4n-3$ assignments of signature $1\dots n\dots 1\dots n\dots 1$.
\end{proof}



\begin{figure}[htbh]
\centerline{\includegraphics[scale=1.6]{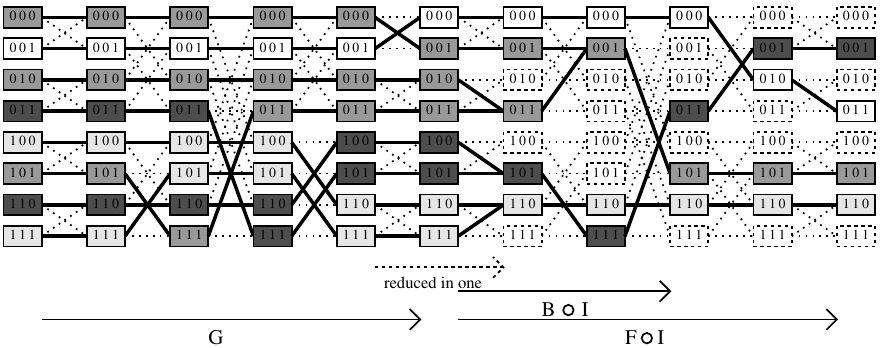}}
\centerline{ \small{ \bf Figure 3.}}
\end{figure}


\begin{example}
{\rm
Figure 3 
gives an example for the construction of this section (on the same mapping as for Figure 2).
The elements of $\{0,1\}^3$ are grouped by the bijection $G$ at stage 6, accordingly with the block-sequence $\bigl[[1,3],[2,2]\bigr]$ induced by $E^{-1}$.
Then, at stage 9 all elements with same final image have been given a same image by $B\circ I$ (which is the mapping detailed in Example \ref{ex:compose}), where $B$ is the first part of the bijection $F$. At last, the second part of the bijection $F$ allows to finalize the mapping $E$.
 Observe that we could have chosen other block-sequences, such as
$\bigl[[3,1],[2,2]\bigr]$ or 
$\bigl[[2,2],[1,3]\bigr]$ for instance, and we would have obtained other in situ programs and routing patterns for the same mapping (see Remark \ref{flexible}).  
}
\end{example}

%


\begin{open problem}
\label{open:4n}
{\rm
A first natural question was to extend the notion of block-sequence and Lemma \ref{lm:partition} to a set $S$ of arbitrary size,  in such a way that the rest of the construction remains valid. An efficient answer has been given recently in \cite{Ga1}, as mentioned in the introduction of this Section. 
Another general and probably demanding question is to study the reach of the flexibility provided by the present construction: either to get general in situ programs of shorter length, as requested by Open problem \ref{open:length}; or, in terms of networks, to get efficient (wide-sense) non-blocking routing methods, i.e. methods to update  dynamically  the routing with respect to updates of the computed mapping. Finding such dynamic non-blocking routing strategies is a main concern of the network field (see \cite{Hw04} for a general background).%
}
\end{open problem}


%
%


\section{Linear mappings on suitable ring powers}
\label{sec:linear}

In this section, we assume that $S$ is given with an algebraic structure:
$S$ is a (non-necessarily finite) quotient of an
Euclidean domain $R$ by an ideal $I$. Classical examples for $S$ are: any
field
(the result of this section for $S$ being a field is easier, since most
technicalities can be skipped), the rings $\bZ/s \bZ$, or ${K[x]}/{(P)}$
for some polynomial $P$ with coefficients in a field $K$.  

Given $S$ and an integer $n$, we consider a \emph{linear} mapping $S^n\rightarrow S^n$,
that is an application from $S^n$ to $S^n$ which is
a linear application with respect to the canonical structure of
$S$-module of $S^n$. 
The results of Section~3 show that $O(n)$
assignments are sufficient to compute such a mapping.  Here, we achieve a
stronger result: the number of required mappings is bounded by $2n-1$,
\emph{and} all intermediary assignments are linear.

In \cite{BuMo04bis}, a similar result is obtained in the particular case
of linear boolean mappings. The paper \cite{RaBo91} achieves this result with an
on-the-fly self routing strategy, again restricted to the linear boolean
case. We note that the more general result we obtain here is not of this
efficient nature, since the actual \emph{computation} of the
decomposition we obtain has a complexity which is of the order of
$O(n^3)$. We insist however on the fact that the in situ decomposition we
provide can be proven to be of length at most $2n-1$, and is available for much more general useful algebraic  structures.

Also, finding an in situ program of a linear mapping using linear assignments
is equivalent to rewriting a matrix as a product of \emph{assignment
matrices} (matrices equal to the identity matrix except on one row).
Theorem~\ref{th:linear} below is proven using this alternate formalism.

We denote $S^*$ the set of invertible elements of $S$.
For convenience we also define the Kronecker symbol $\delta_i^j$, which is
defined for two integers $i$ and $j$ as being $1$ for $i=j$, and $0$
otherwise.

\begin{lemma}
\label{lem:gcdthingy}
Let $x_1,\ldots,x_n$ be coprime elements of $R$. Let $i_0\in[1\ldots n]$.
There
exists multipliers $\lambda_1,\ldots,\lambda_n$ such that
$\lambda_{i_0}=1$, and $\sum_i\lambda_ix_i\in S^*$.%
\end{lemma}

\begin{proof} 
By assumption, the index $i_0$ is fixed. Without loss of generality we
may safely assume that $i_0=1$.  In virtue of the Chinese Remainder
Theorem, it suffices to define the multipliers
$\lambda_1,\ldots,\lambda_n$ separately modulo each prime power $p^v$
dividing $I$, thus it is also valid to restrict to the case where the
ideal $I$ is generated by a prime power $p^v$.

Now we distinguish two cases. In the first case, $x_1$ is coprime to $p$,
and thus coprime to $I$. We thus choose $\lambda_1=1$, and 
$\lambda_i=0$ for all $i>1$. Then $\lambda_1x_1=x_1$ is in $(R/I)^*$. In
the second case, $x_1$ is divisible by $p$. Since by assumption, the
$x_i$'s are coprime, there exists an integer $i_1$ such that
$x_{i_1}$ is coprime to
$p$. Therefore $x_{i_0}+x_{i_1}$ is coprime to $p$, hence we may set
$\lambda_1=\lambda_{i_1}=1$, and $\lambda_i=0$ for all other indices $i$
(in other words, we may write $\lambda_i=\delta_i^1+\delta_i^{i_0}$).
We thus have  $\sum_i\lambda_ix_i=x_{i_0}+x_{i_1}$, which is coprime to
$p$, whence in $(R/I)^*$ as well.
%
\end{proof}

\begin{corollary}
\label{lem:gcdthingy-cor}
Let $x_1,\ldots,x_n$ be elements of $R$, and $g=\gcd(x_1,\ldots,x_n)$. Let $i_0\in[1\ldots n]$.
There
exists multipliers $\lambda_1,\ldots,\lambda_n$ such that
$\lambda_{i_0}=1$, and $\sum_i\lambda_ix_i\in gS^*$.%
\end{corollary}
\begin{proof}
This is a trivial application of Lemma~\ref{lem:gcdthingy} to
$(x_1/g,\ldots,x_n/g)$.
\end{proof}





\begin{theorem} \label{th:linear}
Every linear mapping $E$ on $S^n$ is computed by an in situ program
of length $2n-1$ and signature $1,2,...,n,{n-1},..., 1$ made of linear assignments.

Furthermore, if $E$ is bijective, then the inverse mapping $E^{-1}$ is
computed by the in situ program defined by the same sequence of assignments in reversed order together with the following transformation:

 \noindent $\Bigl[x_i:=a\cdot x_i + f(x_1,\ldots,x_{i-1},x_{i+1},\ldots,x_n)\Bigr]$ 
 
 \hfill $\mapsto\ \ \ \ \Bigl[x_i:=a^{-1}\cdot\bigl(x_i-f(x_1,\ldots,x_{i-1},x_{i+1},\ldots,x_n)\bigr)\Bigr].$
\end{theorem}

\begin{proof}
The proof proceeds by induction. Let $k$ be an integer, and let $M$ be a matrix representing a
linear mapping $E$ on $S^n$  which leaves the first $k-1$ variables
unchanged.
In other words, the first $k-1$ rows of $M$ equal those of the identity matrix.
The matrix which defines the input linear mapping $E$
satisfies this property with $k=1$, thus our induction initiates at $k=1$
with the matrix defining $E$.
We explore
the possibility of rewriting $M$ as a product $L_kM'R_k$, where the first
$k$ rows of $M'$ match those of the identity matrix. 

Let $g$ be the greatest common divisor of (arbitrary representatives in
$R$ of) the coefficients of column $k$ in $M$.
A favourable situation is when $m_{k,k}$ is in $gS^*$.  Should this not be
the case, let us see how we can transform the matrix to reach this
situation unconditionally.  Assume then for a moment that
$m_{k,k}\notin gS^*$.  Lemma~\ref{lem:gcdthingy} gives
multipliers $\lambda_1, \ldots, \lambda_n$ such that
$\sum_\ell\lambda_\ell m_{\ell,k}\in gS^*$, with the additional
constraint that $\lambda_{k}=1$. Let us now denote by $T$ the $n\times n$
matrix which differs from the identity matrix only at row $k$, and whose
coefficients in row number $k$ are defined by
$t_{k,j}=\lambda_j$.  Clearly $T$ is an invertible assignment matrix, and
the product $TM$ has a coefficient at position $(k,k)$ which is in
$gS^*$.

Now assume $m_{k,k}\in gS^*$. Let $G$ be the diagonal matrix having
$G_{k,k}=g$ as the only diagonal entry not equal to $1$.  Let
$M''=MG^{-1}$ ($M''$ has coefficients in $R$ because $g$ is the g.c.d. of
column $k$). We have $m''_{k,k}\in S^*$.  We form an assignment matrix $U$
which differs from the identity matrix only at row $k$, and whose
coefficients in row number $k$ are exactly the coefficients of the $k$-th
row of $M''$.
The matrix $U$ is then an invertible assignment matrix (its determinant is
$m''_{k,k}$). The $k$ first rows of the matrix $M'=M''U^{-1}$ match the
$k$ first rows of $I_n$, and we have $M=T^{-1}\times M'\times (UG)$. Our
goal is therefore reached with $L_k=T^{-1}$ and $R_k=UG$.

Repeating the procedure, our input matrix is rewritten as a product
$$L_1L_2\ldots L_{n-1}R_n\ldots R_1,$$ where all matrices are assignment
matrices. No left multiplier $L_n$
is needed for the last step, since the g.c.d. of one single element is
equal to the element itself.
Finally,
the determinant of $M$ is invertible if and only if
all the matrices $R_k$ are invertible, 
hence the reversibility for bijective mappings.%
\end{proof}

\begin{corollary} \label{cor:matrix}
Every square matrix of size $n$ on $S$ (quotient of an
Euclidean domain by an ideal) is the product of $2n-1$ 
assignment matrices (equal to the identity matrix except on one row).
%
%
\end{corollary}

\begin{remark}
{\rm
We digress briefly on the computational complexity of \emph{building} the
in situ programs for the linear mappings considered here. The matrix
operations performed here all have complexity $O(n^2)$ because of the
special shape of the assignment matrices. Therefore, the overall
computational complexity of the decomposition is $O(n^3)$.
}
\end{remark}

\renewcommand{\arraystretch}{1}

\def\mtwo#1#2#3#4{\ensuremath{\left(\begin{array}{rr}#1& #2\\#3
&#4\end{array}\right)}}
\begin{example}
{\rm
The procedure in the proof of Theorem \ref{th:linear} can be illustrated by a small example. Assume we
want to decompose the mapping in $\bZ/12\bZ$ given by the matrix
$$E=\mtwo4564.$$
Let $M$ be this matrix. The g.c.d of column $1$ in $M$ is $\gcd(4,6)=2$,
which is not invertible modulo $12$.
We therefore firstly use Corollary~\ref{lem:gcdthingy-cor} to make $\gcd(4,6)=2$
appear in the top left coefficient. We use the relation
\begin{align*}
1*(4/2)+1*(6/2)&=5\in({\bZ}/{12\bZ})^*,\\
1*4+1*6&=10\in2({\bZ}/{12\bZ})^*.
\end{align*}
We take therefore $\lambda_1=\lambda_2=1$ and the matrix $T=\mtwo1101$ gives
$$M'=TM=\mtwo1101\mtwo4564=\mtwo{10}964.$$
The common divisor $2$ of column $1$ can then be set aside. Let $G=\mtwo2001$. We have
$M''=M'G^{-1}=TMG^{-1}=\mtwo5934$. Now let $U=\mtwo5901$ reproduce
the first row of $M''$. We have $M''U^{-1}\equiv\mtwo1031$, exploiting
the fact that because $5$ is invertible modulo $12$, $U$ is an
invertible matrix modulo $12$. This
eventually unfolds as the following factorization of $M$:
\begin{gather*}
\mtwo1101M =\mtwo5934\mtwo2001
\equiv\mtwo1031\mtwo5901\mtwo2001.\\
M =\mtwo1{-1}01\mtwo1031\mtwo{10}901.
\end{gather*}

\bigskip


This corresponds to the following sequence of assignments:
\vspace{-3mm}

\begin{align*}
x_1 &:= 10x_1 + 9x_2; & x_2 &:= 3x_1 + x_2; & x_1 &:= x_1 - x_2.
\end{align*}

}
\end{example}

\smallskip

\begin{open problem}
\label{open:alg}
{\rm
As mentioned in the introduction, natural subsequent questions are the following. Let $S=\{0,1\}$ be the binary field, and $E:S^n\rightarrow S^n$ be a mapping the components of which are polynomials of degree at most~$k$.
\begin{enumerate}
\item Does there always exist an in situ program of $E$ the assignments of which are polynomials of degree at most $k$?
\item If so, what is the maximum (over all such $E$) of the minimum (over all such in situ programs of $E$) number of assignments in the program? 
\end{enumerate}

The following example tested on a computer shows that, in contrast to the linear case or the boolean bijective case developed earlier in the paper, this bound has to be strictly larger than $2n-1$. For the mapping
$$E\ :\ (x_1,x_2,x_3)\ \longrightarrow\ 
(\ x_2 x_3,\ x_1 x_3,\ x_1 x_2\ ),$$
the following in situ program of $E$ has the smallest possible number of assignments using degree-2 polynomials:

\hskip 3cm
\vbox{
$x_1 := x_2 + x_2 x_3 + x_1;$

$x_2 := x_3 + x_1 + x_2;$

$x_3 := x_3 + x_2 + x_1 x_2;$

$x_1 := x_3 + x_2 x_3 + x_1 x_3;$

$x_2 := x_3 + x_2 x_3 + x_1 x_3;$

$x_3 := x_3 + x_2 x_3 + x_1 x_3.$
}
}
\end{open problem}

\bigskip

\noindent{\bf Conclusion.}
To conclude, further work can consist in applications: for instance the context of
computations modulo e.g. $\bZ/2^{64}\bZ$ is close to the concern of
integer arithmetic with machine words. Theorem~\ref{th:linear} shows that we can obtain a short sequence for computing linear mappings on such data.
About the in situ approach of this computation, a  question may arise, though, as to whether the constants appearing
 in the computation defeat the claim that the computation avoids
 temporaries. In fact,  such constants can be directly embodied in the
 code, so that they contribute exclusively to the code size and not to its
 requested variable data size (or number of registers). 
Further work can also consist in improving bounds: for instance, it has been claimed very recently in \cite{Ga2} that the tight bound is $\lfloor 3.n/2\rfloor$ linear assignments to compute linear mappings when $S$ is the field $\bZ/q\bZ$ for a prime power $q$ (see also Open problem \ref{open:length}).
Further work can also consist in dynamic routing for a network approcah (see \cite{Hw04} and
Open problem~\ref{open:4n}), or  in  algebraic generalizations (such as proposed by Open problem~\ref{open:alg})...
\bigskip

\noindent {\bf Acknowledgements.}
We are grateful to Jean-Paul Delahaye  for presenting our in situ approach of computation in \cite{De11} and giving to us the idea of Proposition~\ref{permut}.
Also, we are grateful to several anonymous referees for useful comments,  
 suggestions and references, that helped notably searching the network literature.





%





\begin{thebibliography}{10}

\bibliographystyle{plain}



\bibitem{Ag83} 
Agrawal, D.P.: 
\newblock
Graph theoretical analysis and design of multistage interconnection networks. 
\newblock
IEEE Trans. Computers {\bf C32} (1983), 637-648.


\bibitem{Be64}
Bene\v s, V.E.: 
\newblock
Optimal Rearrangeable Multistage Connecting Networks. 
\newblock
Bell System Technical J. {\bf 43} (1964), 1641-1656.



\bibitem{BeFoJM88}
Bermond J.C., Fourneau J.M. and Jean-Marie A.: 
\newblock
A graph theoretical approach to equivalence of multi-stage interconnection networks.
\newblock
Discrete Applied Maths {\bf 22} (1988), 201-214.


%

%

\bibitem{Bu96}
Burckel, S.:
\newblock
Closed Iterative Calculus.
\newblock
Theoretical Computer Science {\bf 158} (1996), 371-378. 

\bibitem{BuMo00}
\newblock Burckel, S., Morillon, M.:  Three Generators for Minimal Writing-Space Computations.
Theoretical Informatics and Applications {\bf 34} (2000) 131--138

\bibitem{BuMo04}
Burckel, S., Morillon, M.:  
\newblock
Quadratic Sequential Computations.
\newblock
Theory of Computing Systems {\bf 37(4)} (2004), 519-525
%

\bibitem{BuMo04bis}
Burckel, S., Morillon, M.:  
\newblock
Sequential Computation of Linear Boolean Mappings.
\newblock
Theoretical Computer Science serie A {\bf  314} (2004), 287-292
%

\bibitem{BuGi08}
Burckel, S., Gioan, E.: 
 \newblock
In situ design of register operations.
\newblock
Proceedings of IEEE Computer Society Annual Symposium on Very Large Scale Integration ISVLSI'08 (2008), 4p.
%
%

\bibitem{BuGiTh09}
Burckel S., Gioan E., Thom\'e E.:
 \newblock
 Mapping Computation with No Memory. 
  \newblock
  Proc. 8th International Conference on Unconventional Computation UC09. LNCS 5715 (2009) 85-97. 



\bibitem{Ga2}
Cameron, P., Fairbairn, B., Gadouleau, M.:
Computing in matrix groups without memory.
ArXiv: 1310.6009 (2013).



%
%



\bibitem{CoTu65}
Cooley, J. W., Tukey, J. W.: 
\newblock
An algorithm for the machine
\newblock
calculation of complex Fourier series. Math. Comput. {\bf 19} (1965), 297-301



\bibitem{De11}
Delahaye, J.-P.: Le calculateur amn\'esique.
Pour La Science {\bf 404} (2011), 88-92.
Available at: \url{http://www2.lifl.fr/~delahaye/pls/2011/208}

\bibitem{Ga1}
Gadouleau, M., Riis, S.:
Memoryless computation: new results, constructions, and extensions.
ArXiv: 1111.6026 (2011).


%

\bibitem{Hw04}
Hwang, F.K.: 
\newblock
The mathematical theory of nonblocking switching networks. 
\newblock
Series on Applied Mathematics vol. 11.
World Scientific Publishing Co., River Edge, NJ, USA, 2004 (2nd ed.).

\bibitem{Le90}
Lea, C.-T.:
\newblock
Multi-$Log_2N$ Networks and Their
Applications in High-speed
Electronic and Photonic
Switching Systems
\newblock
IEEE Transactions on Communications {\bf 38(10)} (1990) 1740-1749

\bibitem{Lei92}
Leighton, F.T.: 
\newblock
Introduction to Parallel Algorithms and Architectures:
Arrays,  Trees,  Hypercubes. 
\newblock
Morgan Kaufmann Publishers, San Francisco, CA, USA, 1992

\bibitem{LiCh99}
Lin, W., Chen, W.-S. E.:
\newblock
Efficient nonblocking multicast communications on baseline networks.
\newblock
Computer Communications 22 (1999), 556-567



\bibitem{Of65}
Ofman, Y. P.:
\newblock
 A universal automaton. 
 \newblock
 Trans. Moscow Math. Soc. {\bf 14}
 (1965), 200-215.

\bibitem{RaBo91}
Raghavendra, C. S. and Boppana, R. V. On self-routing in Bene\v s and
s
shuffle-exchange networks. IEEE Trans. on Computers, 40(9):1057-1064,
1991.

\bibitem{Th78}
Thompson, C. D.:
\newblock
Generalized connection network for parallel processor
interconnection. 
\newblock
IEEE Trans. Comput. {\bf 27
} (1978), 1119-1124 








\end{thebibliography}
\end{document}